\documentclass[a4paper]{article}
\usepackage[utf8]{inputenc}
\usepackage{amsthm,amsmath,amssymb,amsfonts}
\usepackage{tikz}
\usepackage{palatino}

\usepackage{hyperref}
\hypersetup{
    colorlinks=true,
    linkcolor=black,
    urlcolor=blue,
    citecolor=black,
    pdftitle={The Constructive mu-calculus: Game Semantics and Non-Wellfounded Proof Systems},
    pdfauthor={Leonardo Pacheco}
}

\theoremstyle{plain}
\newtheorem{theorem}{Theorem}
\newtheorem{lemma}[theorem]{Lemma}
\newtheorem{proposition}[theorem]{Proposition}

\theoremstyle{definition}
\newtheorem{definition}{Definition}
\newtheorem{remark}{Remark}

\usepackage{amsmath,amssymb,amsfonts}
\usepackage{tikz}
\usepackage{multicol}
\usepackage{bussproofs}
\usepackage{xcolor}

\newcommand{\ck}{\mathsf{CK}}
\newcommand{\muck}{\mathsf{lab}\mu\mathsf{CK}_\omega}
\newcommand{\muik}{\mathsf{lab}\mu\mathsf{IK}_\omega}
\newcommand{\mugk}{\mathsf{lab}\mu\mathsf{GK}_\omega}

\newcommand{\tuple}[1]{{\langle #1 \rangle}}
\newcommand{\verifier}{\mathsf{V}}
\newcommand{\refuter}{\mathsf{R}}
\newcommand{\role}{\mathsf{Q}}
\newcommand{\dualrole}{\bar{\mathsf{Q}}}

\newcommand{\sig}[1]{{\mathrm{sig}^\mathsf{I}\langle #1 \rangle}}
\newcommand{\sub}[1]{{\mathrm{Sub}(#1)}}

\newcommand{\R}{\mathbf{R}}

\definecolor{myblue}{HTML}{C6E0FF}
\newcommand{\principal}[1]{\colorbox{myblue}{$#1$}}
\newcommand{\minor}[1]{\colorbox{myblue}{$#1$}}

\title{The Constructive \texorpdfstring{$\mu$}{mu}-calculus:\\ Game Semantics and Non-Wellfounded Proof Systems}
\author{Leonardo Pacheco \\
{\small Institute of Sciece Tokyo, Japan} \\
{\small \texttt{leonardovpacheco@gmail.com}}}
\date{}

\begin{document}
\maketitle

\begin{abstract}
    We study a variant of the modal $\mu$-calculus based on the constructive modal logic $\mathsf{CK}$.
    We define game semantics for the constructive $\mu$-calculus and prove its equivalence to the birelational Kripke semantics.
    We then use the game semantics to prove the soundness and completeness of a fully-labeled non-wellfounded proof system for it.
    At last, we briefly describe how to adapt the game semantics and proof system to the $\mu$-calculus over other non-classical modal logics.
\end{abstract}

\section{Introduction}
We define a constructive variant of the modal $\mu$-calculus by adding least and greatest fixed-point operators to the constructive modal logic $\ck$.
We then define birelational and Kripke game semantics for the constructive $\mu$-calculus and prove their equivalence.
We combine fully-labeled sequents and non-wellfounded proof systems to define $\muck$.
We prove that $\muck$ is sound and weakly complete with respect to $\ck$-models using our game semantics.
To illustrate the flexibility of our methods, we describe how to adapt $\muck$ to proof systems $\muik$ and $\mugk$ sound and weakly complete with respect to the non-classical modal logics $\mathsf{IK}$ and $\mathsf{GL}$, respectively.

The modal $\mu$-calculus was defined by Kozen~\cite{kozen1983results} by adding least and greatest fixed-point operators to $\mathsf{K}$.
The $\mu$-formulas are famously hard to understand as they allow for entangled least and greatest fixed-point operators.
Game semantics for the $\mu$-calculus they give a more intuitive interpretation of the $\mu$-formulas; they are also a powerful tool to work with.
See~\cite{bradfield2018mucalculus,gradel2003automata,lenzi2010recent} for more information on the classical $\mu$-calculus.

Kozen also defined a Hilbert-style proof system $\mathsf{\mu K}$ for the classical $\mu$-calculus, whose completeness was first proved by Walukiewicz~\cite{walukiewicz1995completeness}.
This result was obtained by an analysis of a non-wellfounded proof system first studied by Niwiński and Walukiewicz~\cite{niwinski1996gamesformucalc}.
Other infinitary proof systems for the $\mu$-calculus have also been studied~\cite{afshari2024demystifying,studer2008mucalculus}.

In constructive modal logic, the duality of the modalities $\Box$ and $\Diamond$ is lost.
These logics have been studied for a long time: some of the first texts on the topic are Fitch \cite{fitch1948intuitionistic} and Prawitz \cite{prawitz1965natural}.
In this paper, we use Mendler and de Paiva's birelational $\ck$-models \cite{mendler2005constructive}.
These models have two relations: one to interpret the intuitionistic part of the language and one to interpret the modal part of the language.
These models are based on those of Wijesekera \cite{wijesekera1990constructive}, but allow worlds where the false proposition $\bot$ holds.

In addition to $\mathsf{CK}$, there are other non-classical variants of modal logic such as intuitionistic modal logic~\cite{simpson1994proof,fischerservi1978finite} and Gödel--Dummett modal logic~\cite{caicedo2010godelmodallogic,caicedo2013finite}.
These are closely related to constructive modal logic
On the axiomatic side, these logics are obtained by adding axioms to constructive modal logic.
On the semantics side, they are obtained by adding further restrictions on $\ck$-models.
It is also common to consider non-classical modal logics containing only one of $\Box$ and $\Diamond$ following~\cite{bozic1984modallogic, ono1977intuitionistic}.
See~\cite{das2023diamonds, degroot2024semantical, simpson1994proof} for more information on the relation between constructive and intuitionistic modal logics.
Note that other non-classical fixed-point logics other than have also appeared in the literature~\cite{afshari2023illfddintlintemplogic,afshari2024mastermodality,balbiani2019intlintemplogic,boudou2022complete,fernandezduque2018intuitionistictemporallogic,jager2016intcommonknowledge,nishimura1982semantical}.
More recently, Afshari \emph{et al.}~\cite{afshari2024demystifying} proved the completeness of a constructive fragment of the classical $\mu$-calculus.

Fully-labeled sequents were first used by Marin \emph{et al.}~\cite{marin2021fullylabeled} to study the intuitionistic modal logic $\mathsf{IK}$.
They consist of labeled formulas $x:\varphi$ along statements $x\preceq y$ and $x R y$ between variables representing worlds.
They were later used to prove the decidability of $\mathsf{IS4}$~\cite{girlando2023intuitionistic} and $\mathsf{IK}$~\cite{girlando2024ik}.
As we will see, they are a flexible tool which can be easily adapted to many non-classical modal logics.

\subparagraph*{Our contribution}
We first define evaluation games for the constructive $\mu$-calculus.
In both classical and constructive games, two players discuss whether a formula is true at a given world of a Kripke model.
In the classical games, it is usual to refer to the players as Verifier and Refuter.
In the constructive games, we will still have two players; but now Verifier and Refuter denote two roles and may switch roles depending on their moves.
This change is necessary because, over constructive semantics, not every formula can be put in negative normal form; and formulas in negative normal form are what allows us to simplify the classical evaluation games.
Furthermore, we also need to consider the semantics of constructive diamonds, which mix universal and existential quantifiers.
Therefore we will need a more delicate argument to prove the equivalence of the semantics in the constructive case.

We then define a fully-labeled non-wellfounded proof system $\muck$.
This is obtained by adding fixed-points to a proof system based on the fully-labeled proof systems of Marin \emph{et al.}~\cite{marin2021fullylabeled} and binested proof system of Gao and Olivetti~\cite{gao2026binested}.
We prove its soundness by showing that any contermodel for a formula $\varphi$ witnesses that $\varphi$ has no proof.
We then prove its completeness by extracting a countermodel $\varphi$ from a failed proof search.

We also sketch how to adapt our results to obtain proof systems for the intuitionistic modal logic $\mathsf{IK}$~\cite{simpson1994proof} the Gödel--Dummett modal logic $\mathsf{GK}$~\cite{caicedo2010godelmodallogic}, and other related modal logics. 
We close the paper with some remarks on future work.

\section{Constructive \texorpdfstring{$\mu$}{mu}-calculus}
\label{sec::preliminaries}
\subsection{Syntax}
The language of the $\mu$-calculus is obtained by adding \emph{least} and \emph{greatest fixed-point operators} $\mu$ and $\nu$ to the language of modal logic.
When defining the fixed-point formulas $\mu X.\varphi$ and $\nu X.\varphi$, we need a syntactical requirement in order to have well-defined semantics.

\begin{definition}
    Fix pairwise disjoint sets $\mathrm{Prop}$ and $\mathrm{Var}$ of propositions and propositional variables, respectively.
    The \emph{$\mu$-formulas} are defined by the following grammar:
    \[
        \varphi := P \mid X \mid \bot \mid \varphi\land\varphi  \mid \varphi\lor\varphi \mid \varphi\to\varphi \mid \Box\varphi \mid \Diamond\varphi \mid  \mu X.\varphi \mid \nu X.\varphi,
    \]
    where $P$ is a proposition, $X$ is a propositional variable occurring only positively in $\varphi$.
    We abbreviate $\varphi \to \bot$ by $\neg\varphi$ and $\bot\to\bot$ by $\top$.
    We also call $\mu$-formulas by formulas.

    The fixed-point formulas $\mu X.\varphi$ and $\nu X.\varphi$ are defined iff $X$ is positive in $\varphi$.
    We classify an occurrence of $X$ as \emph{positive} or \emph{negative} in a given formula by structural induction:
    \begin{itemize}
        \item $X$ is positive and negative in $P$ and in $\bot$;
        \item $X$ is positive in $X$;
        \item if $Y\neq X$, $X$ is positive and negative in $Y$;
        \item if $X$ is positive (negative) in $\varphi$ and $\psi$, then $X$ is positive (negative) in $\varphi\land \psi$, $\varphi\lor \psi$, $\Box\varphi$, and $\Diamond\varphi$;
        \item if $X$ is negative (positive) in $\varphi$ and positive (negative) in $\psi$, then $X$ is positive (negative) in $\varphi\to \psi$;
        \item $X$ is positive and negative in $\mu X.\varphi$ and $\nu X.\varphi$.
    \end{itemize}
\end{definition}

Denote the set of \emph{subformulas} of $\varphi$ by $\sub{\varphi}$.
We use $\eta$ to denote either $\mu$ or $\nu$.
An occurrence of a variable $X$ in a formula $\varphi$ is \emph{bound} iff it in the scope of a fixed-point operator $\eta X$.
An occurrence of $X$ is \emph{free} iff it is not bound.
We treat free occurrences of variable symbols as propositional symbols.
A formula $\varphi$ is a \emph{sentence} iff it has no free variables.
An occurrence of $X$ in $\varphi$ is \emph{guarded} iff it is in the scope of some modality $\Box$ or $\Diamond$.
A formula $\varphi$ is \emph{guarded} iff, for all $\eta X.\psi\in\sub{\varphi}$, $X$ is guarded in $\psi$.
A formula $\varphi$ is \emph{well-bounded} iff, for all variables $X$ occurring bounded in $\varphi$, $X$ occurs only once and there is only one fixed-point operator $\eta X$ in $\varphi$.
A formula is \emph{well-named} iff it is guarded and well-bounded.
Note that every formula is equivalent to a well-named formula.
If $\varphi$ is a well-named formula and $\eta X.\psi\in\sub{\psi}$, denote by $\psi_X$ the formula $\psi$ which is bound by the fixed-point operator $\eta X$.

\subsection{Kripke Semantics}
We consider the birelational semantics first defined by Mendler and de Paiva \cite{mendler2005constructive}:
\begin{definition}
    A $\ck$-model\footnote{The $\ck$-models are sound and complete semantics for the constructive modal logic $\ck$ of Mendler and de Paiva \cite{mendler2005constructive}.}
    is a tuple $M=\tuple{W, W^\bot ,\preceq, R, V}$ where:
    \begin{itemize}
        \item $W$ is a non-empty set of \emph{possible worlds}; 
        \item $W^\bot\subseteq W$ is the set of \emph{fallible worlds}; 
        \item $\preceq$ is a reflexive and transitive binary relation over $W$; 
        \item $R$ is a binary relation over $W$; and 
        \item $V:\mathrm{Prop}\to \mathcal{P}(W)$ is a \emph{valuation function}.
    \end{itemize}
    We call $\preceq$ the \emph{intuitionistic relation} and $R$ the \emph{modal relation}.
    We require that, if $w \preceq v$ and $w\in V(P)$, then $v\in V(P)$; and that $W^\bot\subseteq V(P)$, for all $P\in\mathrm{Prop}$.
    We also require that if $W^\bot$ is closed under $\preceq$ and $R$.
    For convenience, we sometimes write $w\succeq v$ for $v\preceq w$.
\end{definition}

When proving the correctness of game semantics, we will need consider models with augmented valuations.
That is, we will treat free variable symbols as a proposition symbols and extend the valuation to them.
Formally, let $M = \tuple{W, W^\bot,\preceq, R,V}$ be a $\ck$-model, $A\subseteq W$ and $X$ be a variable symbol; the \emph{augmented $\ck$-model} $M[X\mapsto A]$ is obtained by setting $V(X) := A$.
We can similarly further augment an augmented model.
Given any $\mu$-formula $\varphi$, we also define $\|\varphi(A)\|^M := \|\varphi(X)\|^{M[X\mapsto A]}$.

Fix a $\ck$-model $M = \tuple{W, W^\bot,\preceq, R, V}$.
Given a $\mu$-formula $\varphi$, define the operator $\Gamma_{\varphi(X)}(A) := \|\varphi(A)\|^M$.
We denote the \emph{composition} of the relations $\preceq$ and $R$ by $\preceq;R$.
Define the valuation of the $\mu$-sentences over $M$ by induction on the structure of the formulas:

\begin{minipage}[t]{.35\textwidth}
    \begin{itemize}
        \item $\|P\|^M = V(P)$;
        \item $\|\bot\|^M = W^\bot$;
    \end{itemize}
\end{minipage}
\begin{minipage}[t]{.55\textwidth}
    \begin{itemize}
        \item $\|\varphi\land\psi\|^M = \|\psi\|^M \cap \|\psi\|^M$;
        \item $\|\varphi\lor\psi\|^M = \|\psi\|^M \cup \|\psi\|^M$;
    \end{itemize}
\end{minipage}

\smallskip

\begin{minipage}[t]{.9\textwidth}
    \begin{itemize}
        \item $\|\varphi\rightarrow\psi\|^M = \{w\mid \text{for all $v$, if } w\preceq v \text{, then } v\in \|\varphi\|^M \text{ implies } v\in\|\psi\|^M \}$;
        \item $\|\Box\varphi\|^M = \{ w \mid \text{for all $v$, if } w \preceq;R v \text{, then }v\in\|\varphi\|^M\}$;
        \item $\|\Diamond\varphi\|^M = \{ w \mid \text{for all $v$, if $w \preceq v$, then there is $v R u$ such that }u\in\|\varphi\|^M\}$;
        \item $\|\mu X.\varphi(X)\|^M$ is the least fixed-point of the operator $\Gamma_{\varphi(X)}$; and
        \item $\|\nu X.\varphi(X)\|^M$ is the greatest fixed-point of the operator $\Gamma_{\varphi(X)}$.
    \end{itemize}
\end{minipage}

\noindent We write $M,w\models\varphi$ when $w\in\|\varphi\|^M$.
We omit the reference to the model $M$ when it is clear from the context and write $w\in\|\varphi\|$ and $w\models\varphi$.
If $\varphi$ is a sentence, we write $\ck\models\varphi$ iff, for all $\ck$-model $M = \tuple{W, W^\bot ,\preceq, R, V}$ and $w\in W$, we have $M,w\models\varphi$.

The valuation of the fixed-point formulas $\mu X.\varphi$ and $\nu X.\varphi$ is well-defined:
\begin{proposition}
    \label{prop::monotoness_of_Gamma_varphi}
    Fix a $\ck$-model $M=\tuple{W, W^\bot,\preceq,R, V}$ and a formula $\varphi(X)$ with $X$ positive.
    Then the operator $\Gamma_{\varphi(X)}$ is monotone.
    Therefore the valuations of the fixed-point formulas $\mu X.\varphi$ and $\nu X.\varphi$ are well-defined.
\end{proposition}
\begin{proof}
    The monotonicity is proved by structural induction on the $\mu$-formulas.
    The existence of least and greatest fixed-points of monotone operators is guaranteed by the Knaster--Tarski Theorem \cite{arnold2001rudiments}.
    For a detailed proof, see the Appendix~\ref{appendix::proof-A}.
\end{proof}

\subsection{Language Extensions}
In our proofs, we will also use \emph{local diamonds} and \emph{approximants} of fixed-point formulas.
Local diamonds allow us to express the truth of a formula in a locally accessible world, they will be useful in both the game semantics and the proof systems.
The $\alpha$th approximant $\mu X^\alpha.\varphi$ of $\mu X.\varphi$ is obtained by applying $\Gamma_\varphi(X)$ to $\emptyset$, $\alpha$ many times; similarly, the approximant $\nu X^\alpha.\varphi$ is obtained by applying $\Gamma_\varphi(X)$ to $W$, $\alpha$ many times.
Approximants will play a central role in the proof of the correctness of the game semantics.

\begin{definition}
    Let $M$ be a $\ck$-model and $\varphi(X)$ be a formula where $X$ is positive.
    Then:
    \[
        \|\hat\Diamond\varphi\|^M := \{ w \mid \text{there is $v$ such that } wRv\text{ and } v\in\|\varphi\|^M\}.
    \]
\end{definition}
\noindent It is easy to show that there is a model $M$ and worlds $w,v$ such that $w\preceq v$, $w\models\hat\Diamond$, and $v\not\models\hat\Diamond$.

\begin{definition}
    Let $M$ be a $\ck$-model and $\varphi(X)$ be a formula where $X$ is positive.
    The \emph{approximants} of $\mu X.\varphi$ and $\mu X.\varphi$ on $M$ are defined by:
    \begin{itemize}
    \item $\|\mu X^0.\varphi\|^M = \emptyset$, $\|\nu X^0.\varphi\|^M = W$;
    \item $\|\mu X^{\alpha+1}.\varphi\|^M = \|\varphi(\|\mu X^\alpha.\varphi\|^M)\|$, $\|\nu X^{\alpha+1}.\varphi\|^M = \|\varphi(\|\nu X^\alpha.\varphi\|^M)\|$, where $\alpha$ is any ordinal; and
    \item $\|\mu X^\lambda.\varphi\|^M = \bigcup_{\alpha<\lambda} \|\mu X^\alpha.\varphi\|^M$, $\|\nu X^\lambda.\varphi\|^M = \bigcap_{\alpha<\lambda} \|\nu X^\alpha.\varphi\|^M$, where $\lambda$ is a limit ordinal.
\end{itemize}
\end{definition}
\noindent Note that, for all $\ck$-model $M$ and $\mu$-sentence $\varphi(X)$ with $X$ positive there are ordinal numbers $\alpha$ and $\beta$ such that $\|\mu X.\varphi\|^M = \|\mu X^\alpha.\varphi\|^M$ and $\|\nu X.\varphi\|^M = \|\nu X^\beta.\varphi\|^M$.

\section{Game Semantics}
\label{sec::game-semantics}
In this section, we define game semantics for the constructive $\mu$-calculus and prove its equivalence to the birelational semantics.
This game semantics is a modification of the game semantics of the classical $\mu$-calculus.

In evaluation games for the classical $\mu$-calculus, the players Verifier and Refuter discuss whether a formula $\varphi$ holds in a world $w$ of a Kripke model $M$.
While in the classical $\mu$-calculus we can suppose formulas use no implications and that negations are applied only to propositional symbols, we cannot do the same in the constructive $\mu$-calculus.
This complicates the games used for the constructive $\mu$-calculus: the players now have the \emph{roles} of Refuter and Verifier, and swap roles when discussing certain formulas.
Furthermore, we also need to consider the fact that constructive diamonds are not purely existential, and so both players are involved in choosing the next move.

\subsection{Definition of Game Semantics}
\label{subsec::def-game-semantics}
Fix a $\ck$-model $M=\tuple{W,W^\bot,\preceq,R,V}$, a world $w\in W$, and a well-named sentence $\varphi$.
In this subsection, we define the \emph{evaluation game} $\mathcal{G}(M,w\models\varphi)$.

The game $\mathcal{G}(M,w\models\varphi)$ has two players: $\mathsf{I}$ and $\mathsf{II}$.
The two players will have the roles of Verifier and Refuter (abbreviated to $\verifier$ and $\refuter$, respectively).
Each player has only one role at any given time, and the players always have different roles.
We usually write ``$\mathsf{I}$ is $\verifier$'' for ``the player $\mathsf{I}$ has the role of $\verifier$'', and similar expressions for the other combination of players and games.
We denote an arbitrary role by $\role$ and the dual role by $\dualrole$; that is, if $\role$ is $\verifier$, then $\dualrole$ is $\refuter$ and \emph{vice versa}.

The game has two types of positions.
The main positions of the game are of the form $\tuple{v, \psi, \role}$ where $v\in W$, $\psi\in\sub{\varphi}$, and $\role$ is a role.
We also have auxiliary positions of the form $\tuple{v, \hat\Diamond\psi, \role}$ where $\Diamond\psi\in\sub{\varphi}$ and $\tuple{v, \psi?\theta, \role}$ where $\psi\to\theta\in\sub{\varphi}$, for all $v\in W$.
In any position, $\role$ is the role currently held by $\mathsf{I}$; the role of $\mathsf{II}$ is $\dualrole$.
Intuitively, at a position $\tuple{v,\psi, \role}$, $\verifier$ tries to show that $v$ satisfies $\psi$ and $\refuter$ tries to show that $v$ does not satisfy $\psi$.

These auxiliary positions make explicit that both players have input on the next main position after $\tuple{v,\Diamond\psi, \role}$ or $\tuple{v,\psi\to\theta, \role}$.
For example, at the position $\tuple{v,\Diamond\psi, \role}$, $\refuter$ first makes a choice of $v'$ such that $v\preceq v'$ and then $\verifier$ picks $v''$ such that $v' R v''$.
By a similar reasoning, we also need auxiliary positions for $\tuple{v,\theta\to\theta', \role}$.

The game begin at the position $\tuple{w,\varphi,\verifier}$, with $\mathsf{I}$ in the role of $\verifier$ and $\mathsf{II}$ in the role of $\refuter$.

Each position is owned by exactly one of the players.
At each turn of the game, the owner of the current position chooses one of the available positions to move to.
The game then continues with the new position.
If no such position is available, the game ends.
The ownership and available plays for each of position are described in Table \ref{table::intuitionistic-evaluation_game}:
\begin{table}[htb]\renewcommand\arraystretch{1.2}
    \caption{Rules of evaluation games for the constructive modal $\mu$-calculus.}
    \begin{multicols}{2}
    \begin{tabular}{c|c}
    \multicolumn{2}{c}{Verifier}\\
    \hline
    Position &  Admissible moves\\

    $\tuple{v,P, \role}$ and $v \not\in V(P)$  &  $\emptyset$ \\

    $\tuple{v,\bot, \role}$ and $v \not\in W^\bot$  &  $\emptyset$ \\

    $\tuple{v, \psi \lor \theta, \role}$ &  $\{\tuple{v, \psi, \role} , \tuple{v,\theta, \role} \}$\\

    $\tuple{v,\psi?\theta, \role}$ & $\{\tuple{v,\psi, \dualrole}, \tuple{v,\theta, \role}\}$ \\

    $\tuple{v, \hat\Diamond\psi, \role}$ & $\{ \tuple{u, \psi, \role} \mid vR u \}$ \\

    $\tuple{v,\mu X.\psi_X, \role}$ &  $\{\tuple{v,\psi_X, \role} \}$ \\

    $\tuple{v,X, \role}$ & $\{ \tuple{v,\mu X.\psi_X, \role} \}$ \\
    \end{tabular} \\
    \begin{tabular}{c|c}
    \multicolumn{2}{c}{Refuter}\\
    \hline
    Position &  Admissible moves\\
    $\tuple{v, P, \role}$ and $v \in V(P)$  &  $\emptyset$ \\

    $\tuple{v, \bot, \role}$ and $v \in W^\bot$  &  $\emptyset$ \\

    $\tuple{v,\psi \land \theta, \role}$ &  $\{ \tuple{v,\psi, \role} , \tuple{v,\theta, \role} \}$ \\

    $\tuple{v,\psi\to\theta, \role}$ & $\{\tuple{u,\psi?\theta, \role} \mid v\preceq v\}$ \\

    $\tuple{v, \Box \psi, \role}$ & $\{ \tuple{u, \psi, \role} \mid  v \preceq;R u\}$ \\

    $\tuple{v, \Diamond \psi, \role}$ & $\{ \tuple{u, \hat\Diamond\psi, \role} \mid v\preceq u \}$ \\

    $\tuple{v,\nu X.\psi_X, \role}$ &  $\{\tuple{v,\psi_X}, \role\}$ \\

    $\tuple{v,X, \role}$ & $\{ \tuple{v,\nu X.\psi_X, \role} \}$ \\
    \end{tabular}
    \end{multicols}
    \label{table::intuitionistic-evaluation_game}
\end{table}

A \emph{run} of the game is a complete sequence of positions which respects the rules above.
That is, a run $\rho$ is a (finite or infinite) sequence $\tuple{v_0,\psi_0,\role_0}, \tuple{v_1,\psi_1,\role_1}, \tuple{v_2,\psi_2,\role_2}, \dots$ of positions such that:
\begin{itemize}
    \item $\tuple{v_0,\psi_0,\role_0}$ is $\tuple{w,\varphi,\verifier}$;
    \item it is possible to play $\tuple{v_{n+1},\psi_{n+1},\role_{n+1}}$ from $\tuple{v_n,\psi_n,\role_n}$; and
    \item if $\rho$ is finite, then the last position in $\rho$ is of the form $\tuple{v,P,\role}$ with $P\in\mathrm{Prop}$ or $\tuple{v,\psi,\role}$ with $v\in W^\bot$.
\end{itemize}

Before defining the winning conditions, we note that the positivity requirement on the fixed-point formulas guarantees that, if $\tuple{v,\eta X.\psi_X, \role}$ and $\tuple{v',\eta X.\psi_X, \role'}$ occur in any run of the game, then $\mathsf{I}$ has the same role in both positions:
\begin{proposition}
    \label{prop::ownership-well-defined}
    Let $\rho$ and $\rho'$ be a runs of the game $\mathcal{G}(M,w\models\varphi)$.
    Suppose $\tuple{v,\eta X. \psi_X, \role}$ occurs in $\rho$ and $\tuple{v',\eta X. \psi_X, \role'}$ occurs in $\rho'$.
    Then $\role = \role'$; that is, $\mathsf{I}$ has the same role at both positions.
\end{proposition}
\begin{proof}
    See Appendix~\ref{appendix::proof-B}.
\end{proof}
We say that the fixed-point formula \emph{$\eta X.\psi_X$ is owned by $\mathsf{I}$} if a position of the form $\tuple{v,\eta X.\psi_X,\role}$ is reachable from the initial position, and either $\role = \verifier$ and $\eta = \nu$, or $\role = \refuter$ and $\eta = \mu$.
The fixed-point formula $\eta X.\psi_X$ is \emph{owned by $\mathsf{II}$} if it is not owned by $\mathsf{I}$.
This is well-defined by Proposition \ref{prop::ownership-well-defined}.

We are now ready to define the \emph{winning conditions} for the game.
When moving from $\tuple{v,X, \role}$ to $\tuple{v,\eta X.\psi_X, \role}$, we say that the fixed-point formula $\eta X.\psi$ was \emph{regenerated}.
Let $\rho$ be a run of the game.
If $\rho$ is finite, then the last position in $\rho$ is of the form $\tuple{v,P,\role}$ with $P\in\mathrm{Prop}$ or $\tuple{v,\psi,\role}$ with $v\in W^\bot$.
The owner of the last position has no available position and loses the game.
If $\rho$ is infinite, let $\eta X.\psi_{X}$ outermost infinitely often regenerated fixed-point formula in the play $\rho$; that is, $\eta X$ is regenerated infinitely often in $\rho$ and, if $\eta' Y$ is regenerated infinitely often in $\rho$, then $\eta' Y.\psi_Y\in\mathrm{Sub}(\eta X.\psi_X)$.
Then $\mathsf{I}$ wins $\rho$ iff $\mathsf{I}$ owns the fixed-point $\eta X$.

A \emph{(positional) strategy} for $\mathsf{I}$ is a function $\sigma$ which, given a position $\tuple{v,\psi, \role}$ owned by $\mathsf{I}$, outputs a position $\tuple{v',\psi', \role}$ where $\mathsf{I}$ can move to, if any such position is available.
$\mathsf{I}$ follows $\sigma$ in the run $\rho = \tuple{v_0,\psi_0,\role_0}, \tuple{v_1,\psi_1,\role_1}, \tuple{v_2,\psi_2,\role_2}, \dots$ if whenever $\tuple{v_n}, \psi_n, \role_n$ is owned by $\mathsf{I}$, then $\tuple{v_{n+1}}, \psi_{n+1}, \role_{n+1} = \sigma(\tuple{v_n}, \psi_n, \role_n)$.
The strategy $\sigma$ is \emph{winning} iff $\mathsf{I}$ wins all possible runs where they follow $\sigma$.
Strategies and winning strategies for $\mathsf{II}$ are defined similarly.

It is not immediate that that one of the players has a winning strategy for a given evaluation game.
The existence of winning strategies is implied by our proof of the equivalence of the birelational Kripke semantics and game semantics for the constructive $\mu$-calculus:
\begin{theorem}
    Fix a $\ck$-model $M=\tuple{W,W^\bot,\preceq,R,V}$, a world $w\in W$, and a well-named sentence $\varphi$.
    Let $\mathcal{G}(M,w\models\varphi)$ be an evaluation game.
    Then only one of $\mathsf{I}$ and $\mathsf{II}$ has a positional winning strategy $\mathcal{G}(M,w\models\varphi)$.
\end{theorem}
\begin{proof}
    Suppose $\sigma$ is a winning strategy for $\mathsf{I}$ and $\tau$ is a winning strategy for $\mathsf{II}$.
    Then the play resulting from the players using $\sigma$ and $\tau$ is winning for both $\mathsf{I}$ and $\mathsf{II}$, which is not possible by the definition of the game.
    So at most one of the players has a winning strategy.

    Now, by the definition of the birelational semantics, either $M,w\models\varphi$ or $M,w\not\models\varphi$.
    Theorem \ref{thm::correctness-game-semantics} provides us strategies in both cases.
    In case $M,w\models\varphi$ holds, $\mathsf{I}$ has a winning strategy; in case $M,w\not\models\varphi$ holds, $\mathsf{II}$ has a winning strategy.
\end{proof}
Note that the existence of winning positional strategies for the evaluation games is equivalent to the existence of positional winning strategies for parity games.
The existence of such strategies was first proved by Emerson and Jutla \cite{emerson1991tree}, and is equivalent to the relatively strong principle of determinacy $\forall n.(\Sigma^0_2)_n$-$\mathsf{Det}$ \cite{leszek2016rabin}.
See \cite{gradel2003automata} for more on the relation between the $\mu$-calculus and parity games.

\subsection{Equivalence of Kripke and Game Semantics}
\label{subsec::correctness-game-semantics}
We now show the equivalence between the constructive $\mu$-calculus' birelational semantics and game semantics.
That is, we will show that $M,w\models\varphi$ iff the player $\mathsf{I}$ has a winning strategy for $\mathcal{G}(M,w\models\varphi)$, and that $M,w\not\models\varphi$ iff the player $\mathsf{II}$ has a winning strategy for $\mathcal{G}(M,w\models\varphi)$.
We briefly sketch the key idea and remark on some technical points before stating the theorem.

Fix an evaluation game $\mathcal{G}(M,w\models\varphi)$.
The position $\tuple{v,\psi, \role}$ is a \emph{true position} iff $M,w\models \psi$, and a \emph{false position} iff $M,w\not\models\psi$.
We will show that, at a true positions, $\verifier$ can always move in a way favorable to themselves.
That is, in a way such that the resulting position is a true position, if the players have not switched roles; or the resulting position is a false position, if the players have switched roles.
On the other hand, $\refuter$ cannot move in any way favorable to themselves.
A similar situation occurs at false positions.
To make the statements above precise, we have to overcome two problems.

First, when considering whether $M, v\models\psi$, the formula $\psi$ might have free variables, and so its valuation might not be well-defined.
We solve this by augmenting $M$ with the intended valuations for the variables occurring in $\varphi$.

Second, we need to consider infinite plays.
Specifically, we need to guarantee that, when starting from a true position, the resulting play is winning for $\mathsf{I}$; and when starting from a false position, the resulting play is winning for $\mathsf{II}$.
To solve this, we assign two signatures to each position; both are finite sequences of ordinals ordered with the lexical order.
They will be the key tool to guarantee the resulting plays are winning for the corresponding players.

\begin{theorem}
    \label{thm::correctness-game-semantics}
    The birelational Kripke semantics and game semantics for $\ck$ are equivalent.
    That is, if $M=\tuple{W, W^\bot,\preceq,R,V}$ is a $\ck$-model, $w\in W$ and $\varphi$ is a well-named $\mu$-sentence, then:
    \begin{itemize}
        \item $\mathsf{I}$ has a winning strategy for $\mathcal{G}(M,w\models\varphi)$ if and only if $M,w\models\varphi$; and
        \item $\mathsf{II}$ has a winning strategy for $\mathcal{G}(M,w\models\varphi)$ if and only if $M,w\not\models\varphi$.
    \end{itemize}
\end{theorem}
\begin{proof}
    See Appendix~\ref{appendix::proof-C}.
\end{proof}

\section{Non-wellfounded proof system}
\label{sec::non-wellfounded-proofs}
In this section, we define the fully-labeled non-wellfounded proof system $\muck$ and then prove its soundness and weak completeness with respect to $\ck$-models. 
This proof system combines labeled-sequents and non-wellfounded proofs and is closely related to our game semantics for $\ck$.

\subsection{Basic Definitions}

Fix a set of \emph{world variables} $\mathrm{Var}'$ disjoint from the set of propositional variables $\mathrm{Var}$.
A \emph{labeled formula} is of the form $x:\varphi$ where $x\in\mathrm{Var}'$ and $\varphi$ is a formula.
A \emph{labeled sentence} is a labeled formula $x:\varphi$ where $\varphi$ is a sentence.
Given a labeled formula $x:\varphi$, we say the variable $x$ labels the formula $\varphi$.
When not ambiguous, we call both propositional and world variables by `variables'.

A \emph{sequent} $S$ is of the form $\R, \Gamma \vdash \Delta$, where $\R$ is a finite set of statements of the form $x\preceq y$ or $xRy$ with $x,y\in\mathrm{Var}$; and $\Gamma$, $\Delta$ are finite sets of labeled formulas.
$\ell(S)$ is the collection of variables occurring in $S$, either on the relation $\R$ or as labels for formulas.
Below, we write ``$x\preceq y$'' and ``$xRy$'' to mean $x\preceq y \in \R$ and $xRy \in \R$, respectively; we shall continue doing this below if there is no risk of confusion.
Given finite sets sets of formulas $\Gamma$ and $\Delta$, and a labeled formula $x:\varphi$, we write ``$\Gamma,\Delta$'' for $\Gamma\cup\Delta$ and ``$\Gamma,x:\varphi$'' for $\Gamma\cup \{x:\varphi\}$. 
Similarly, if $\R$ and $\R'$ are finite sets of statements, then we write ``$\R,\R'$'' for $\R\cup\R'$.

\begin{definition}
We say that a finite tree of sequents $T$ is a \emph{pre-proof} if the following conditions hold:
\begin{enumerate}
\item every leaf of $T$ is an \textit{axiom}, \emph{i.e.}, of one of the following forms:
\begin{multicols}{2}
    \begin{prooftree}
        \AxiomC{}
        \LeftLabel{$\bot$l}
        \UnaryInfC{$\R,\Gamma, x:\bot \vdash \Delta, x:P$}
    \end{prooftree}
    \begin{prooftree}
        \AxiomC{}
        \LeftLabel{$\mathsf{id}$}
        \UnaryInfC{$\R,\Gamma, x:P \vdash \Delta, x:P$}
    \end{prooftree}
\end{multicols}
where $P\in\mathrm{Prop}$.
\item each non-leaf node in $T$ and the edges above it corresponds to one of the following inference rules:
\setlength{\columnsep}{1cm}
{ \scriptsize
\begin{multicols}{2}
    \begin{prooftree}
        \AxiomC{$\R, x \preceq y,\Gamma, x:\varphi, y:\varphi \vdash \Delta$}
        \LeftLabel{$\preceq$-pres}
        \UnaryInfC{$\R, x\preceq y, \Gamma, x:\varphi\vdash \Delta$}
    \end{prooftree}

    \begin{prooftree}
        \AxiomC{$\R,\Gamma, x:\varphi\land \psi, \minor{x:\varphi}, \minor{x:\psi} \vdash \Delta$}
        \LeftLabel{$\land$l}
        \UnaryInfC{$\R,\Gamma, \principal{x:\varphi\land \psi}, \vdash \Delta$}
    \end{prooftree}
    \begin{prooftree}
        \AxiomC{$\R,\Gamma, x:\varphi\lor \psi, \minor{x:\varphi} \vdash \Delta$}
        \AxiomC{$\R,\Gamma, x:\varphi\lor \psi, \minor{x:\psi} \vdash \Delta$}
        \LeftLabel{$\lor$l}
        \BinaryInfC{$\R,\Gamma, \principal{x:\varphi\lor \psi} \vdash \Delta$}
    \end{prooftree}
    \begin{prooftree}
        \AxiomC{$\R,\Gamma, x:\varphi\to \psi \vdash \Delta, \minor{x:\varphi}$}
        \AxiomC{$\R,\Gamma, x:\varphi\to \psi, \minor{x:\psi}\vdash \Delta$}
        \LeftLabel{$\to$l}
        \BinaryInfC{$\R,\Gamma, \principal{x:\varphi\to \psi} \vdash \Delta$}
    \end{prooftree}
    \begin{prooftree}
        \AxiomC{$\R, \Gamma, x:\Box \varphi, \{\minor{y:\varphi} \mid xRy\} \vdash \Delta$}
        \LeftLabel{$\Box$l}
        \UnaryInfC{$\R, \Gamma, \principal{x:\Box \varphi} \vdash \Delta$}
    \end{prooftree}
    \begin{prooftree}
        \AxiomC{$\R, x R y,\Gamma, x:\Diamond \varphi, \minor{y:\varphi} \vdash \Delta$}
        \RightLabel{($y$ fresh)}
        \LeftLabel{$\Diamond$l}
        \UnaryInfC{$\R,\Gamma, \principal{x:\Diamond \varphi} \vdash \Delta$}
    \end{prooftree}


    \begin{prooftree}
        \AxiomC{$\R, x \preceq y, y\preceq z, x\preceq z, \Gamma \vdash \Delta$}
        \LeftLabel{$\preceq$-trans}
        \UnaryInfC{$\R, x\preceq y, y \preceq z, \Gamma \vdash \Delta$}
    \end{prooftree}
    \begin{prooftree}
        \AxiomC{$\R,\Gamma \vdash \Delta, x:\varphi\land \psi, \minor{x:\varphi}$}
        \AxiomC{$\R,\Gamma \vdash \Delta, x:\varphi\land \psi, \minor{x:\psi}$}
        \LeftLabel{$\land$r}
        \BinaryInfC{$\R,\Gamma \vdash \Delta, \principal{x:\varphi\land \psi}$}
    \end{prooftree}
    \begin{prooftree}
        \AxiomC{$\R,\Gamma \vdash \Delta, x:\varphi\lor \psi, \minor{x:\varphi}, \minor{x:\psi}$}
        \LeftLabel{$\lor$r}
        \UnaryInfC{$\R,\Gamma\vdash \Delta, \principal{x:\varphi\lor \psi}$}
    \end{prooftree}
    \begin{prooftree}
        \AxiomC{$\R, x \preceq y, \Gamma, y:\varphi \vdash \Delta, x:\varphi\to \psi, \minor{y:\psi}$}
        \LeftLabel{$\to$r}
        \RightLabel{($y$ fresh)}
        \UnaryInfC{$\R, \Gamma \vdash \Delta, \principal{x:\varphi\to \psi}$}
    \end{prooftree}
    \begin{prooftree}
        \AxiomC{$\R, x\preceq y, yRz , \Gamma \vdash \Delta, x:\Box \varphi, \minor{y:\varphi}$}
        \RightLabel{($y,z$ fresh)}
        \LeftLabel{$\Box$r}
        \UnaryInfC{$\R, \Gamma \vdash \Delta, \principal{x:\Box \varphi}$}
    \end{prooftree}
    \begin{prooftree}
        \AxiomC{$\R, x\preceq y, \Gamma \vdash \Delta , x:\Diamond \varphi, \minor{y:\hat\Diamond \varphi}$}
        \LeftLabel{$\Diamond$r}
        \RightLabel{($y$ fresh)}
        \UnaryInfC{$\R, \Gamma \vdash \Delta, \principal{x:\Diamond \varphi}$}
    \end{prooftree}
\end{multicols}
    \begin{prooftree}
        \AxiomC{$\R,\Gamma \vdash \Delta , x:\hat\Diamond \varphi, \{\minor{y:\varphi} \mid xRy\}$}
        \LeftLabel{$\hat\Diamond$r}
        \UnaryInfC{$\R, \Gamma \vdash \Delta, \principal{x:\hat\Diamond \varphi}$}
    \end{prooftree}
    
\begin{multicols}{2}
    \begin{prooftree}
        \AxiomC{$\R, \Gamma, x:\eta X.\varphi_X, \minor{x:\varphi_X} \vdash \Delta$}
        \LeftLabel{$\eta$l}
        \UnaryInfC{$\R, \Gamma, \principal{x:\eta X.\varphi_X} \vdash \Delta$}
    \end{prooftree}
    \begin{prooftree}
        \AxiomC{$\R, \Gamma, x:X, \minor{x:\varphi_X} \vdash \Delta$}
        \LeftLabel{$\mathsf{regen}$-l}
        \UnaryInfC{$\R, \Gamma, \principal{x:X} \vdash \Delta$}
    \end{prooftree}
    
    
    \begin{prooftree}
        \AxiomC{$\R, \Gamma \vdash \Delta, x:\eta X.\varphi_X, \minor{x:\varphi_X}$}
        \LeftLabel{$\eta$r}
        \UnaryInfC{$\R, \Gamma \vdash \Delta, \principal{x:\eta X.\varphi}$}
    \end{prooftree}
    \begin{prooftree}
        \AxiomC{$\R, \Gamma \vdash \Delta, x:X, \minor{x:\varphi_X}$}
        \LeftLabel{$\mathsf{regen}$-r}
        \UnaryInfC{$\R, \Gamma \vdash \Delta, \principal{x:X}$}
    \end{prooftree}
\end{multicols}
}
\end{enumerate}
where in the rules $\to$r, $\Box$r, $\Diamond$l, and $\Diamond$r, we require that the \emph{eigenvariables} $y$ and $z$ do not occur in the conclusions.
Given any inference rule, except $\preceq$-pres and $\preceq$-trans, acts on \emph{minor formulas} in the premise(s) and a \emph{principal formula} in the conclusion, these are the \colorbox{myblue}{highlighted} formulas in the definition above.
\end{definition}

\begin{remark}
    We comment on the differences between the fully-labeled proof systems for intuitionistic modal logic in \cite{marin2021fullylabeled} and our.
    Since we allow fallible worlds, in $\bot$l, we require that $x:\bot$ appears on the left and $x:\varphi$ appears on the right for some $\varphi$.
    Since constructive diamonds have a universal-existential definition, we also modify the rules for diamonds.
    $\Diamond$r generates a formula using a local diamond $\hat\Diamond$ on the right; the original rule for diamonds on the right is now used for $\hat\Diamond$.
    We do not modify $\Diamond$r, since the universal part of diamonds is already covered by $\preceq$-pres.
\end{remark}

If $T$ be a pre-proof, we denote by $[T]$ the set of infinite branches of a tree.
A \emph{trace} along a path $\{S_i\}_{i<\omega}\in[T]$ is a sequence of labeled formulas $\{x_i:\varphi_i\}_{i<\omega}$ such that either: $x_i:\varphi_i$ and $x_{i+1}:\varphi_{i+1}$ are the same; or $x_i:\varphi_i$ is the principal formula and $x_{i+1}:\varphi_{i+1}$ is a minor formula of the inference in $T$ where $S_i$ is the conclusion and $S_{i+1}$ is a premise.

We say the formula $\eta X.\psi_X$ is \emph{left-regenerated infinitely often} in $\rho$ if there is an infinite subset $I\subseteq \omega$ such that $x_i:X\in \Gamma_n$ and $x_{i+1}:\eta X.\varphi\in \Gamma_{n+1}$ for all $i\in I$.
Similarly, a formula $\eta X.\psi_X$ is \emph{right-regenerated infinitely often} in $\rho$ if there is an infinite subset $I\subseteq \omega$ such that $x_i:X\in \Delta_n$ and $x_{i+1}:\eta X.\varphi\in \Delta_{n+1}$ for all $i\in I$.
We say $\eta X.\psi_X$ is \emph{regenerated infinitely often} in $\rho$ if it is right- or left-regenerated infinitely often.

Let $\rho$ be a trace and $\eta X.\psi_X$ be the outermost infinitely often regenerated formula in $\rho$, that is, is $\eta X'.\psi'_{X'}$ is regenerated infinitely often in $\rho$ then $\eta X'.\psi'_{X'} \in \eta X.\psi_X$.
A trace is \emph{progressing} if, either $\eta = \nu$ and $\eta X.\psi_X$ is right-regenerated infinitely often, or $\eta = \mu$ and $\eta X.\psi_X$ is left-regenerated infinitely often.

\begin{definition}
    A pre-proof $T$ is a \emph{$\muck$-proof} iff all infinite path $\vec{S}\in [T]$ has a progressing trace. 
    A well-named labeled sentence $x:\varphi$ is \emph{$\muck$-provable} when there is a $\muck$-proof whose root is $x:\varphi$.
    We write $\muck\vdash\varphi$ when there is $x$ such that $x:\varphi$ is \emph{$\muck$-provable}.
\end{definition}

The following theorem will follow from Lemmas~\ref{lem::soundness} and~\ref{lem::completeness} below:
\begin{theorem}
    \label{thm::sound-and-complete}
    $\muck$ is sound and (weakly) complete with respect to $\ck$-models.
    That is, for all well-named sentence $\varphi$, $\muck\vdash\varphi$ if and only if $\ck\models\varphi$.
\end{theorem}

\subsection{Soundness}
\label{subsec::soundness}
To prove the soundness of $\muck$, we show, if $x:\varphi$ has a refutation, then this refutation witnesses that any pre-proof whose root is $x:\varphi$ is not a proof.
We first make the definition of refutation precise:
\begin{definition}
    Let $\varphi$ be a sentence, $S = \R,\Gamma\vdash\Delta$ be a sequent where all formulas are subformulas of $\varphi$, $M = \tuple{W, W^\bot, \preceq, R, V}$ be a $\ck$-model, and $\tau$ be a strategy for $\mathsf{II}$ on $\mathcal{G}(M,w\models\varphi)$.
    A function $r:\ell(S) \to W$ is a \emph{$\tau$-refutation} of $S$ on $M$ iff
    \begin{itemize}
        \item for all $x,y\in \ell(S)$, $xRy\in \R$ implies $r(x)Rr(y)$ and $x\preceq y\in \R$ implies $r(x)\preceq r(y)$;
        \item if $x:\psi\in\Gamma$, then $\tau$ is a winning strategy for $\mathsf{II}$ at $\tuple{r(x), \psi, \refuter}$; and
        \item if $x:\psi\in\Delta$, then $\tau$ is a winning strategy for $\mathsf{II}$ at $\tuple{r(x), \psi, \verifier}$.
    \end{itemize}
\end{definition}

\begin{lemma}
    \label{lem::soundness}
    For all sentence $\varphi$, $\muck\vdash\varphi$ implies $\ck\models\varphi$.
\end{lemma}
\begin{proof}
    Let $x:\varphi$ be a labeled formula and $T$ be a pre-$\muck$-proof of $x:\varphi$.
    Suppose that $M = \tuple{W, W^\bot, \preceq,R,V}$ is a $\ck$ model such that $M,w\not\models\varphi$ for some $w\in W$.
    By Theorem \ref{thm::correctness-game-semantics}, $\mathsf{II}$ wins $\mathcal{G}(M,w\models\varphi)$ via a strategy $\tau$.

    We build a path $\tuple{S_i}_{i\in\omega}$ and a sequence of function $\tuple{r_i}_{i\in\omega}$ where, for all $i\in\omega$, $r_i\subseteq r_{i+1}$ and $r_i$ is a $\tau$-refutation of $S_i$ on $M$.

    Let $S_0$ be $\vdash x:\varphi$ and $r_0(x):= w$. It is immediate that $r_0$ is a refutation of $S_0$.

    Now, suppose $S_i$ and $r_i$ are defined.
    As $r_i$ is a $\tau$-refutation of $S_i$, we can use $\tau$ to pick a premise $S_{i+1}$ of the inference rule in $T$ whose conclusion is $S$.
    Depending on the inference rule, we use $\tau$ to extend $r_i$ into a $\tau$-refutation $r_{i+1}$ defined on the fresh variables in $S_{i+1}$ (or set $r_{i+1} := r_i$ if there is no such variable).
    Due to limited space, we describe only a few representative cases:
    \begin{itemize}
        \item Suppose $S_i$ the conclusion of a $\preceq$-pres, $\preceq$-trans, $\land$l, $\lor$r, $\Box$l, $\hat\Diamond$r, $\eta$r, $\eta$l, regen-r, or regen-l inference rule in $T$. We define $S_{i+1}$ as the only premise of this inference rule and define $r_{i+1} := r_i$.

        \item Suppose $S_i$ the conclusion of an $\to$l inference rule in $T$ whose principal formula is $y:\psi\to\chi$.
        We have that $\tau$ is a winning strategy at $\tuple{e(y), \psi\to\chi, \refuter}$, and so it is also a winning strategy at $\tuple{e(y), \psi?\chi, \refuter}$.
        If $\tau(\tuple{e(y), \psi?\chi, \refuter}) = \tuple{e(y), \theta, \refuter}$, we define $S_{i+1}$ to be the premise whose minor formula is $y:\theta$.
        In either case, define $r_{i+1}:= r_i$.
        A similar argument holds when $S_i$ is the conclusion of a $\lor$l and $\land$r inference rule.
        
        \item Suppose $S_i$ the conclusion of an $\to$r inference rule in $T$ whose principal formula is $y:\psi\to\chi$.
        Again, we define $S_{i+1}$ as the only premise of this inference rule.
        Let is the fresh variable in $S_{i+1}$, 
        If $x\in \mathrm{Var}(S_i)$, set $r_{i+1}(z) = r_i(z)$.
        If $z$ is the fresh variable in $\mathrm{Var}(S_{i+1})\setminus\mathrm{Var}(S_{i+1})$, set $r_{i+1}(z) = v$, where $v\in W$ is such that $\tau(r(y),\psi\to\chi,\verifier) = \tuple{v,\psi?\chi,\verifier}$.
        As $r_i$ is a $\tau$-refutation of $S_i$, so $r_{i+1}$ is a $\tau$-refutation of $S_{i+1}$.
        A similar argument holds when $S_i$ is the conclusion of a $\Box$r, $\Diamond$l, or $\Diamond$r inference rule; note that, in the case of $\Box$r rule, the domain of $r_{i+1}$ will contain two elements not in the domain of $r_i$.
    \end{itemize}

    By the above construction, each trace in the path $\tuple{S_i}_{i\in\omega}$ is corresponds to a play in $\mathcal{G}(M,w\models\varphi)$ following $\tau$. 
    As $\tau$ is a winning strategy for $\mathsf{II}$, no trace in $\tuple{S_i}_{i\in\omega}$ is progressing.
    Therefore $T$ is not a $\muck$-proof.
\end{proof}

\subsection{Completeness}
\label{sec::completeness}

Before proving the completeness for $\muck$, we define a notion of saturation for sequents:
\begin{definition}
    We say the sequent $\R, \Gamma \vdash \Delta$ is \emph{saturated} iff the following conditions hold:
    \begin{enumerate}
        \item if $x\preceq y$ and $y\preceq z$, then $x\preceq z$;
        \item if $xRy$ and $yRz$, then $xRz$;
        \item if $x\preceq y$ and $x:\varphi\in\Gamma$, then $y:\varphi\in\Gamma$;
        \item if $x:\varphi\land\psi\in\Gamma$, then $x:\varphi\in\Gamma$ and $x:\psi\in\Gamma$;
        \item if $x:\varphi\land\psi\in\Delta$, then $x:\varphi\in\Delta$ or $x:\psi\in\Delta$;
        \item if $x:\varphi\lor\psi\in\Gamma$, then $x:\varphi\in\Gamma$ or $x:\psi\in\Gamma$;
        \item if $x:\varphi\lor\psi\in\Delta$, then $x:\varphi\in\Delta$ and $x:\psi\in\Delta$;
        \item if $x:\varphi\to\psi\in\Gamma$, then $x:\varphi\in\Delta$ or $x:\psi\in\Gamma$;
        \item if $x:\varphi\to\psi\in\Delta$, then there is a variable $y$ such that $y:\varphi\in\Gamma$ and $y:\psi\in\Delta$;
        \item if $x:\Box\varphi\in\Gamma$, then, for all variable $y$, $xRy$ implies $y:\varphi\in\Gamma$;
        \item if $x:\Box\varphi\in\Delta$, then there are variables $y,z$ such that $x\preceq y$, $y R z$, and $z:\varphi\in\Delta$;
        \item if $x:\Diamond\varphi\in\Gamma$, then there is a variable $y$ such that $xRy$ and $y:\varphi\in\Gamma$;
        \item if $x:\Diamond\varphi\in\Delta$, then there is a variable $y$ such that $xRy$ and $y:\hat\varphi\in\Delta$;
        \item if $x:\Diamond\hat\varphi\in\Delta$, then, for all variable $y$, $xRy$ implies $y:\varphi\in\Delta$;
        \item if $x:\eta X.\varphi_X\in\Gamma$, then $x:\varphi_X\in\Gamma$;
        \item if $x:\eta X.\varphi_X\in\Delta$, then $x:\varphi_X\in\Delta$;
        \item if $x:X\in\Gamma$ $x:\eta X.\varphi_X\in\Gamma$; and
        \item if $x:X\in\Delta$ $x:\eta X.\varphi_X\in\Delta$.
    \end{enumerate}
    Note that there cannot be any formula of the form $x:\hat\Diamond\varphi$ in $\Gamma$.
\end{definition}

To prove completeness, we do a proof search which will give us either a $\muck$-proof or will allow us to define a countermodel.
Note that all inference rules are \emph{invertible}, that is, the premises of each rule can be obtained from the conclusion by an application of the weakening rule:
\begin{prooftree}
    \AxiomC{$\R, \Gamma \vdash \Delta$}
    \LeftLabel{$\mathsf{wk}$}
    \UnaryInfC{$\R, \R', \Gamma,\Gamma' \vdash \Delta,\Delta'$}
\end{prooftree}
We do not include the weakening rule explicitly in our system since it is implicit in the other rules.
\begin{lemma}
    \label{lem::completeness}
    For all $\mu$-sentence $\varphi$, $\ck\models\varphi$ implies $\muck\vdash\varphi$.
\end{lemma}
\begin{proof}
    Fix a labeled sentence $x_0:\varphi_0$.
    We build a pre-proof $T$ by defining a sequence of finite trees $\{T_i\}_{i\in\omega}$ with $T_i\subseteq T_{i+1}$ for all $i\in\omega$.
    Let $T_0$ be a tree whose only node is its root, $x_0:\varphi_0$.

    If $T_i$ is defined, we define $T_{i+1}$ extends as follows.
    Let $S = \R,\Gamma\vdash\Delta$ be a leaf of $T_i$ which are not axioms nor saturated.
    If $y:\psi$ is a formula witnessing that $S$ is not saturated, we add above $S$ the premises of an inference rule whose principal formula is $x:\psi$.
    For example, if $y:\Box\chi$ witnesses $S$ is not saturated, we add above $S$ the sequent $\R,\Gamma,\{ z:\psi \mid yR z \}\vdash\Delta$.
    To ensure all formulas witnessing non-saturation are eventually acted upon, we give priority to formulas witnessing sequents below $S$ are non-saturated.
    Repeat this process for all leaves of $T_i$ which are not axioms nor saturated.
    
    Let $T = \bigcup_{i\in\omega}T_i$ be the union of all $T_i$.
    By construction, $T$ is a pre-$\muck$-proof.
    If $T$ is also a $\muck$-proof, we are done.
    Otherwise, we will extract a countermodel $M$ from $T$.

    There are two possibilities we need to consider.
    If $T$ has a saturated leaf which is not an axiom, we let $S = \R,\Gamma\vdash\Delta$ be such a leaf.
    Otherwise, let $\tuple{S_i}_{i\in\omega}$ be a path of $T$ with no progressing trace; in this case, let $S$ be $\bigcup_{i\in\omega}S_i$.
    That is, $S$ is $\R,\Gamma\vdash\Gamma = \bigcup_{i\in\omega}\R_i, \bigcup_{i\in\omega}\Gamma_i \vdash \bigcup_{i\in\omega}\Delta_i$, where $S_i = \R_i,\Gamma_i\vdash\Delta_i$.
    Note that, by construction, $S_i$ is saturated.\footnote{Remember that we officially consider only finite sequents; this is a harmless abuse of notation.}
    We now define $M = \tuple{W, W^\bot, \preceq,R, V}$ as follows:
    \begin{itemize}
        \item $W := \ell(S)$;
        \item $W^\bot := \{ x\in W \mid x:\bot \in \Gamma \}$;
        \item $\preceq := \{\tuple{x,y} \in W^2 \mid x\preceq y\in \R\}$;
        \item $R := \{\tuple{x,y} \in W^2 \mid x R y\in \R\}$; and
        \item $V(P) := \{ x\in W \mid x:P\in \Gamma \}$.
    \end{itemize}
    It is straightforward to show that $M$ is a $\ck$-model.

    We define a strategy $\tau$ for $\mathsf{II}$ in $\mathcal{G}(M,x\models\varphi)$ in a way such that $r(x) := x$ is a $\tau$-refutation of $S$ on $M$.
    At an position owned by $\mathsf{II}$, $\tau$ will move to a position respecting the saturation rules.
    For example, if $y:\Box\psi\in \Delta$, at a position of the form $\tuple{y,\Box\chi,\verifier}$, $\mathsf{II}$ moves to $\tuple{z, \chi, \verifier}$ where $x\preceq y$, $y R z$, and $z:\varphi\in\Delta$.
    Such moves will always exist for any formula by the saturation of $S$.
    
    If $r$ is not a $\tau$-refutation of $S$ on $M$, we get a winning run $\rho$ for $\mathsf{I}$ from either a node $\tuple{y,\psi,\verifier}$ with $y:\psi\in\Delta$ or $\tuple{y,\psi,\refuter}$ with $y:\psi\in\Delta$.
    If $\rho$ is finite and its last position is $\tuple{y,P,\role}$, we get that $y:P\in \Gamma\cap\Delta$.
    Depending on the construction of $S$, either $S$ is an axiom in $T$, or $S$ is the union of sequents $\{S_i\}_{i\in\omega}$ where some $S_i$ is an axiom.
    In both cases, we get a contradiction.
    If $\rho$ is infinite, $S$ must be the union of a sequence of sequents $\tuple{S_i}_{i\in\omega}$, as the well-naming requirement on formulas require progressing traces to have infinitely many variables.
    Here, we can extract a progressing trace in $\tuple{S_i}_{i\in\omega}$, a contradiction.

    As $r$ is a $\tau$-refutation of $S$ on $M$, we get $M,x_0\not\models\varphi_0$, that is, $\varphi_0$ is not valid.
\end{proof}

\section{\texorpdfstring{$\muik$}{muIK} and \texorpdfstring{$\mugk$}{muGK}}
\label{sec::extensions-of-ck}
To illustrate the flexibility of our methods, we define fully-labeled non-wellfounded proof systems $\muik$ and $\mugk$ for the $\mu$-calculus based on models of the logics $\mathsf{IK}$ and $\mathsf{GK}$, respectively.
For more details on $\mathsf{IK}$ refer to \cite{das2023diamonds,simpson1994proof}; for more details on $\mathsf{GK}$, refer to \cite{caicedo2010godelmodallogic}.

We begin by defining $\mathsf{IK}$- and $\mathsf{GK}$-models based on $\ck$-models.
\begin{definition}
    Let $M = \tuple{W, W^\bot, \preceq, R, V}$ be a $\ck$-model.
    We say $M$ is an \emph{$\mathsf{IK}$-model} iff:
    \begin{itemize}
        \item $W^\bot = \emptyset$;
        \item $M$ is forward confluent: if $w \preceq w'$ and $w R v$, then there is $v'$ such that $w' R v'$ and $v\preceq v'$;
        \item $M$ is backward confluent: if $w R v \preceq v'$, then there is $w'$ such that $w\preceq w' R v'$.
    \end{itemize}
    We say an $\mathsf{IK}$-model is a \emph{$\mathsf{GK}$-model} iff it is locally linear: if $w\preceq v$ and $w\preceq u$, then either $v\preceq u$ or $u\preceq v$.
    We write $\mathsf{IK}\models\varphi$ to mean that, for all $\mathsf{IK}$-model $M = \tuple{W, \emptyset, \preceq, R, V}$ and all $w\in W$, we have $M,w\models\varphi$.
    We define $\mathsf{GK}\models\varphi$ similarly.
\end{definition}

The following inference rules will be added to $\muck$ to define the $\muik$ and $\mugk$:
{\footnotesize
\begin{prooftree}
    \AxiomC{}
    \LeftLabel{$\mathsf{id}_{\Diamond \mathrm{N}}$}
    \UnaryInfC{$\R, \Gamma, x:\bot \vdash \Delta$}
\end{prooftree}
\begin{multicols}{2}
\begin{prooftree}
    \AxiomC{$\R, x\preceq x', x R y, y\preceq y', x' R y', \Gamma \vdash \Delta$}
    \LeftLabel{$\mathsf{fwd}$}
    \RightLabel{($y'$ fresh)}
    \UnaryInfC{$\R, x\preceq x', x R y, \Gamma, \vdash \Delta$}
\end{prooftree}
\begin{prooftree}
    \AxiomC{$\R, x R y, y \preceq y', x\preceq x', x' R y',\Gamma \vdash \Delta$}
    \LeftLabel{$\mathsf{bwd}$}
    \RightLabel{($x'$ fresh)}
    \UnaryInfC{$\R, x R y, y \preceq y',\Gamma, \vdash \Delta$}
\end{prooftree}    
\end{multicols}
\begin{prooftree}
    \AxiomC{$\R, x\preceq y, x\preceq z, y\preceq z, \Gamma \vdash \Delta$}
    \AxiomC{$\R, x\preceq y, x\preceq z, z\preceq y, \Gamma \vdash \Delta$}
    \LeftLabel{$\mathsf{lin}$}
    \BinaryInfC{$\R, x\preceq y, x\preceq z, \Gamma, \vdash \Delta$}
\end{prooftree}
}
These inferences correspond to the additional constraints on $\mathsf{IK}$- and $\mathsf{GK}$-models.

\begin{definition}
    \emph{$\muik$-proofs} defined as $\muck$-proofs, but we allow inferences rules $\mathsf{id}_{\Diamond \mathrm{N}}$, $\mathsf{fwd}$, and $\mathsf{bwd}$ to be used in addition to the inferences of $\muck$.
    Similarly, \emph{$\mugk$-proofs} are defined as $\muik$-proofs, but we further allow the inference rule $\mathsf{lin}$ to be used.
\end{definition}

As we proved Theorem \ref{thm::sound-and-complete}, we can prove:
\begin{theorem}
    $\muik$ and $\mugk$ are sound and (weakly) complete with respect to $\mathsf{IK}$- and $\mathsf{GK}$-models, respectively.
\end{theorem}

We note that this argument can also be extended and adapted to other non-classical modal logics, although one must be careful when doing so.
For example, to define an analogous proof system over $\mathsf{CK4}$ models~\cite{balbiani2021constructive}, we need to add not only inference rules for the reflexivity and transivity of the modal relation, but also the rule $\mathsf{bwd}$. This happens because $\mathsf{CK4}$-models are \emph{backward confluent} $\ck$-models where the relation is reflexive and transitive.
Similar care must be taken with the logic $\mathsf{CS5}$, which coincides with $\mathsf{IS5}$~\cite{pacheco2024collapsing}; so we need to add $\mathrm{N}_\Diamond$, $\mathsf{fwd}$, and $\mathsf{bwd}$ in addition to rules for the reflexity, transitivity, and symmetry for the modal relation.

\section{Future Work}
\label{sec::future-work}
In future work, we plan to define a cyclic version of $\muck$ similar to the one in~\cite{pacheco2025cigl} and then adapt the loop checking methods of \cite{girlando2023intuitionistic,girlando2024ik} to obtain decidability results.

We close this paper with a(n infinite collection of) questions.
Let $\mathsf{H\mu CK}$ be Hilbert-style system obtained by adding fixed-point axioms and rules to the axioms of $\ck$.\footnote{See~\cite{das2023diamonds} for details.}
Do $\muck$ and $\mathsf{H\mu CK}$ prove the same theorems? While one direction follows from Theorem~\ref{thm::sound-and-complete}, the other is not so.
An analogous question can be asked for $\muik$, $\mugk$, and the $\mu$-calculus based on other non-classical modal logics.

\bibliographystyle{alphaurl}
\bibliography{intuitionistic-mu-calculus}

\appendix
\section{Proof of Proposition \ref{prop::monotoness_of_Gamma_varphi}}
\label{appendix::proof-A}
    Fix a model $M$ and sets $A\subseteq B \subseteq W$.
    We prove that if $X$ is positive in $\varphi$ then $\|\varphi(A)\|\subseteq \|\varphi(B)\|$; and that if $X$ is negative in $\varphi$, then $\|\varphi(B)\|\subseteq \|\varphi(A)\|$.
    We will focus on the positive case, as the negative case is dual.

    The proof is by structural induction on the $\mu$-formulas.
    The cases of formulas of the form $P$, $X$, $Y$, $\varphi\land\psi$, and $\varphi\lor\psi$ follow by direct calculations.
    The case for formulas of the form $\eta X.\varphi$ is trivial, as $X$ is not free in $\eta X.\varphi$.

    We now prove the proposition for formulas of the form $\varphi\to\psi$.
    Suppose $X$ is positive in $\varphi\to \psi$, then $X$ is positive in $\psi$ and negative in $\varphi$.
    Therefore:
    \begin{align*}
      w\in\|(\varphi\to\psi)(A)\| \iff& \text{for all $v$, if $w \preceq v$ and }v\in \|\varphi(A)\|\text{, then } v\in\|\psi(A)\| \\
      \implies& \text{for all $v$, if $w \preceq v$ and }v\in \|\varphi(B)\|\text{, then } v\in\|\psi(B)\| \\
      \iff&w\in  \|(\varphi\to\psi)(B)\|.
    \end{align*}
    The case for formulas of the form $\neg\varphi$ is similar.

    Finally, we prove the proposition for formulas of the form $\Box\varphi$.
    Suppose $X$ is positive in $\Box\varphi$, then $X$ is positive in $\varphi$.
    Therefore:
    \begin{align*}
        w\in\|\Box\varphi(A)\| \iff& \text{for all $v, u$ such that $w\preceq v R u$, }u\in \|\varphi(A)\| \\
        \implies& \text{for all $v, u$ such that $w\preceq v R u$, }u\in \|\varphi(B)\| \\
        \iff& w\in \|\Box\varphi(B)\|.
    \end{align*}
    The proof for formulas of the form $\Diamond\varphi$ is similar.

\section{Proof of Proposition \ref{prop::ownership-well-defined}}
\label{appendix::proof-B}
    If $\tuple{v,\eta X. \psi_X, \role}$ and $\tuple{v',\eta X. \psi_X, \role'}$ occur in the same run $\rho$, then $\role$ and $\role'$ coincide by the positivity of $X$ in $\psi_X$: it implies that the players must swap roles an even number of times between these two positions.

    Now, let $\tuple{v,\eta X. \psi_X, \role}$ and $\tuple{v',\eta X. \psi_X, \role'}$ be the first occurrence of positions with the formula $\eta X.\psi_X$ in $\rho$ and $\rho'$, respectively.
    The well-namedness of $\varphi$ implies that there is only one occurrence of $\eta X$ in $\varphi$.
    This fact along with the positivity of $X$ implies that the number of times the players switch roles to get to $\tuple{v,\eta X. \psi_X, \role}$ and $\tuple{v',\eta X. \psi_X, \role'}$ must have the same parity.

\section{Proof of Theorem \ref{thm::correctness-game-semantics}}
\label{appendix::proof-C}
    We prove that, if $M,w\models\varphi$, then $\mathsf{I}$ has a winning strategy for $\mathcal{G}(M,w\models\varphi)$ and, if $M,w\not\models\varphi$, then $\mathsf{II}$ has a winning strategy for $\mathcal{G}(M,w\models\varphi)$.
    This is sufficient to prove the theorem since the two players cannot both have a winning strategy for $\mathcal{G}(M,w\models\varphi)$ and since one of $M,w\models\varphi$ or $M,w\not\models\varphi$ always holds.

    Suppose $w\models \varphi$.
    We will assign to each main position $\tuple{v,\psi, \role}$ of the game an ordinal signature $\sig{v,\psi, \role}$.
    We show $\mathsf{I}$ is always able to control the truth of the positions in the evaluation game $\mathcal{G}(M,w\models\varphi)$ and move in a way that the signature is eventually constant.

    Before we define $\mathsf{I}$-signatures, we enumerate the fixed-point subformulas of $\varphi$ in non-increasing size:
    \[
        \eta_1 X_1.\psi_1, \eta_2 X_2.\psi_2, \dots, \eta_n X_n.\psi_n.
    \]
    That is, we enumerate the fixed-point subformulas of $\varphi$ in a way such that, if $i < j$, then $\eta_{i} X_{i}.\psi_{i}\not\in \mathrm{Sub}(\eta_j X_j.\psi_j)$; and, if $\eta_{i} X_{i}.\psi_{i}\in \mathrm{Sub}(\eta_j X_j.\psi_j)$, then $j \leq i$.
    We also enumerate the fixed-point subformulas of $\varphi$ which are owned by $\mathsf{I}$ in non-increasing size:
    \[
        \eta_1' Y_1.\chi_1, \eta_2' Y_2.\chi_2, \dots, \eta_m' Y_m.\chi_m.
    \]

    An \emph{$\mathsf{I}$-signature} $r = \tuple{r(1), \dots, r(m)}$ is a sequence of $m$ ordinals.
    Denote by $r(k)$ the $k$th component of $r$.
    Write $r =_k r'$ iff the first $k$ components of $r$ are identical.
    Order the signatures by the lexicographical order: $r<r'$ iff there is $k\in\{0,\dots, m-1\}$ such that $r =_k r'$ and $r(k+1)<r'(k+1)$.
    The lexicographical order is a well-ordering of the signatures.

    Remember that the \emph{augmented Kripke model} $M[X\mapsto A]$ is obtained by setting $V(X) := A$, where $M = \tuple{W, W^\bot,\preceq, R,V}$ is a Kripke model, $A\subseteq W$ and $X$ is a variable symbol.
    We want to evaluate subformulas of $\varphi$ where some $X_1,\dots, X_n$ occur free, so we augment $M$ with the correct valuations of these variables:
    \begin{align*}
        M_0     &:= V; \\
        M_{i} &:= M_{i-1}[X_{i}\mapsto \|\eta_{i} X_{i}.\psi_{i}\|^{M_i}]\text{, for $1\leq i \leq n$.}
    \end{align*}
    By the choice of our enumeration, $\eta_{i} X_{i}.\psi_{i}$ does not contain free occurrences of $X_{i},\dots, X_n$, and so $M_i$ is well-defined.

    The $\alpha$-th approximant $\mu X^\alpha.\varphi$ of $\mu X.\varphi$ is obtained by applying $\Gamma_\varphi(X)$ to $\emptyset$, $\alpha$ many times; and that the approximant $\nu X^\alpha.\varphi$ is obtained by applying $\Gamma_\varphi(X)$ to $W$, $\alpha$ many times.
    We define models $M_n^r$ where the variables $Y_j$ owned by $\mathsf{I}$ are assigned their $r(j)$-th approximant $\|\eta_j^{r(j)} Y_j.\chi_j\|$, and variables owned by $\mathsf{II}$ receive their correct value.
    Formally, given a signature $r$, we define augmented models $M^r_0, \dots, M^r_n$ by
    \begin{align*}
        M_0^r     &:= V; \\
        M_{i}^r &:= \left\{
        \begin{array}{ll}
            M_i[X_{i}\mapsto \|\eta_j' Y_j^{r(j)}.\chi_j\|^{M_i^r}], & \text{if $ X_{i} = Y_j$}; \\
            M_i[X_{i}\mapsto \|\eta_{i} X_{i}.\psi_{i}\|^{M_i}], & \text{if there is no $j$ such that $X_{i} = Y_j$}
        \end{array} \right. \text{, for $1\leq i \leq n$.}
    \end{align*}

    If $M_n,v\models\psi$, we call $\tuple{v,\psi, \role}$ a \emph{true position}; if $M_n,v \not\models\psi$, we call $\tuple{v,\psi, \role}$ a \emph{false position}.
    Now, if $\tuple{v,\psi, \role}$ a true position, then there is a least signature $r$ such that $M_n^r,v\models\psi$.
    Similarly, if $\tuple{v,\psi, \role}$ a false position, then there is a least signature $r$ such that $M_n^r,v\not\models\psi$.
    Denote these signatures by $\sig{v,\psi, \role}$.

    We will define a strategy $\sigma$ for $\mathsf{I}$ which guarantees that when the players are at $\tuple{v,\psi, \role}$, $v\models\psi$ if $\mathsf{I}$ is in the role of $\verifier$, and $v\not\models\psi$ if $\mathsf{I}$ is in the role of $\refuter$.
    Furthermore, $\mathsf{II}$ cannot move in ways the signature increases and $\mathsf{I}$'s moves are eventually non-increasing.
    The only time the signature may increase is when regenerating some fixed-poin formula $\eta'_jY_j.\chi_j$, but in this case the first $j-1$ positions of the signature are not modified.

    We will also have if a position is reachable when $\mathsf{I}$ follows the strategy $\sigma$, then it is a true positions when $\role = \verifier$ and a false position when $\role=\refuter$.
    Remember that the game starts on the position $\tuple{w,\varphi,\verifier}$ and we assumed that $M,w\models\varphi$ holds, so this is true for the initial position of the game.

    We define $\mathsf{I}$'s strategy as follows:
    \begin{itemize}
        \item Suppose the game is at the position $\tuple{v,\theta_1\lor\theta_2, \role}$.
        If $\role = \verifier$ and $\tuple{v,\theta_1\lor\theta_2, \verifier}$ is a true position; then $\mathsf{I}$ moves to $\tuple{v,\theta_i, \verifier}$ such that $M_n^\sig{v,\theta_1\lor\theta_2, \verifier}, v\models \theta_i$, with $i\in\{1,2\}$.
        By the definition of the signatures, $\sig{v,\theta_1\lor\theta_2, \verifier} = \sig{v,\theta_i, \verifier}$.
        If $\role = \refuter$ and $\tuple{v,\theta_1\lor\theta_2,\refuter}$ is a false position; then $M_n^\sig{v,\theta_1\lor\theta_2,\refuter}, v\not\models \theta_i$ and $\sig{v,\theta_1\lor\theta_2,\refuter} \geq \sig{v,\theta_i,\refuter}$ for all $i\in\{1,2\}$. So whichever way $\mathsf{II}$ moves, the next position is false and the signature is non-increasing.

        \item Suppose the game is at the position $\tuple{v,\theta_1\land\theta_2,\role}$.
        If $\role = \verifier$ and $\tuple{v,\theta_1\land\theta_2,\verifier}$ is a true position; then $M_n^\sig{v,\theta_1\land\theta_2,\verifier}, v\models \theta_i$ and $\sig{v,\theta_1\land\theta_2,\verifier} \geq \sig{v,\theta_i}$ for all $i\in\{1,2\}$.
        So whichever way $\mathsf{II}$ moves, the next position is true and the signature is non-increasing.
        If $\role = \refuter$ and $\tuple{v,\theta_1\land\theta_2,\refuter}$ is a false position; then $\mathsf{I}$ moves to $\tuple{v,\theta_i,\refuter}$ such that $M_n^\sig{v,\theta_1\land\theta_2,\refuter}, v\not\models \theta_i$ and $\sig{v,\theta_1\land\theta_2,\refuter} = \sig{v,\theta_i,\refuter}$, with $i\in\{1,2\}$.

        \item Suppose the game is at the position $\tuple{v,\Diamond\theta,\role}$.
        If $\role = \verifier$ and $\tuple{v,\Diamond\theta,\verifier}$ is a true position; for all move $\tuple{\tuple{v'},\theta,\verifier}$ of $\mathsf{II}$, $\mathsf{I}$ can move to some $\tuple{v'',\theta,\verifier}$ such that $M_n^\sig{v,\Diamond\theta,\verifier}, v''\models \theta$.
        By the definition of the signatures, $\sig{v,\Diamond\theta,\verifier} \geq \sig{v'',\theta,\verifier}$.
        If $\role = \refuter$ and $\tuple{v,\Diamond\theta,\refuter}$ is a false position; $\mathsf{I}$ moves to a position $\tuple{\tuple{v'},\theta,\refuter}$ such that all answers $\tuple{v'',\theta,\refuter}$ by $\mathsf{II}$ are false positions.
        Furthermore, $M_n^\sig{v,\Diamond\theta,\refuter}, v''\not\models \theta$ and $\sig{v,\Diamond\theta,\refuter} \geq \sig{v'',\theta,\refuter}$ for all such $v''$.

        \item Suppose the game is at the position $\tuple{v,\Box\theta,\role}$.
        If $\role = \verifier$ and $\tuple{v,\Box\theta,\verifier}$ is a true position; for all move $\tuple{v', \theta,\verifier}$ of $\mathsf{II}$, we have $M_n^\sig{v,\Box\theta,\verifier}, v'\models \theta$.
        By the definition of the signatures, $\sig{v,\Box\theta,\verifier} \geq \sig{v'',\theta,\verifier}$.
        If $\role = \refuter$ and $\tuple{v,\Box\theta,\refuter}$ is a false position; $\mathsf{I}$ moves to a false position $\tuple{v',\theta,\refuter}$ with $v\preceq;R v'$.
        Furthermore, $M_n^\sig{v,\Box\theta,\refuter}, v'\not\models \theta$ and $\sig{v,\Box\theta,\refuter} \geq \sig{v',\theta}$.

        \item Suppose the game is at the position $\tuple{v,\theta_1\to\theta_2,\role}$.
        If $\role = \verifier$ and $\tuple{v,\theta_1\to\theta_2,\verifier}$ is a true position.
        After $\mathsf{II}$ moves to $\tuple{v',\theta_1?\theta_2,\verifier}$, $\mathsf{I}$ moves to $\tuple{v',\theta_2,\verifier}$ if it is a true position.
        Otherwise, $\mathsf{I}$ moves to $\tuple{v',\theta_1,\refuter}$; in this case, $\tuple{v',\theta_1,\refuter}$ is a false position.
        Either way, $\sig{v,\theta_1\to\theta_2,\verifier} \geq \sig{v',\theta_i,\role'}$.
        If $\role = \refuter$ and $\tuple{v,\theta_1\to\theta_2,\refuter}$ is a false position; $\mathsf{I}$ moves to a position $\tuple{v',\theta_1?\theta_2,\refuter}$ such that $\tuple{v',\theta_1\verifier}$ is a true position  and $\tuple{v',\theta_2,\refuter}$ is a false position.
        Any answer of $\mathsf{II}$ satisfies our requirements.

        \item Suppose the game is at $\tuple{v,\eta_j' Y_j.\chi_j,\role}$ or at $\tuple{v,Y_j,\role}$, then the owner of the position must move to $\tuple{v,\chi_j,\role}$.
        If there there is $j$ such that $X_i = Y_j$, then $\sig{v,\eta_j' Y_j.\chi_j} =_{j-1} \sig{v,Y_j} =_{j+1} \sig{w,\chi_j}$ and $\sig{w,Y_j}(j) > \sig{w,\chi_j}(j)$.
        If there is no $j$ such that $X_i = Y_j$, then $\sig{w,\eta_{i+1} X_{i+1}.\theta_{i+1}} = \sig{w,X_i} = \sig{w,\theta_i}$.
    \end{itemize}

    On finite runs, $\mathsf{I}$ wins by the construction of the strategy $\sigma$: $\mathsf{I}$ is $\mathsf{V}$ at a true position of the form $\tuple{v,P,\role}$ reachable following $\sigma$.
    Similarly, $\mathsf{I}$ is $\mathsf{R}$ at a false positions of the form $\tuple{v,P,\role}$.
    Also, $\mathsf{I}$ is $\mathsf{V}$ at true positions $\tuple{v,\psi,\role}$ where $v\in W^\bot$.

    Now, consider an infinite run $\tuple{w_0,\varphi_0,\role_0}, \tuple{w_1, \varphi_1,\role_1}, \dots$ where $\mathsf{I}$ follows $\sigma$, let $i$ be the smallest number in $\{1,\dots, n\}$ such that $\eta_i X_i$ is an infinitely often regenerated fixed-point operator.
    Suppose there is $j\in\{1,\dots, m\}$ such that $X_i = Y_j$ for a contradiction.
    Let $k_1,k_2,\dots$ be the positions where $Y_j$ occur; that is, all the positions $\tuple{w_{k_l}, \psi,\role} = \tuple{w_{k_l}, Y_j,\role}$.
    Without loss of generality, we suppose that for all $i'<i$ no $X_{i'}$ is regenerated after the $k_1$th position of the run.
    The move from $\tuple{w_{k_i}, Y_j,\role}$ to $\tuple{w_{k_{i}+1}, \chi_j,\role}$ causes a strict decrease in the signature.
    The other moves between $k_i+1$ and $k_{i+1}$ cannot cancel this decrease, since either the signature does not change or one of the first $i$ positions of the signature is reduced.
    Therefore the sequence of signatures
    \[
        \sig{w_{k_1}, Y_j,Q_{k_1}}, \sig{w_{k_2}, Y_j,Q_{k_2}}, \sig{w_{k_3}, Y_j,Q_{k_3}}, \dots
    \]
    is strictly decreasing.
    This is a contradiction, as the signatures are well-ordered.
    Therefore there is no $j$ such that $X_i = Y_j$, and so $\mathsf{I}$ wins the run.

    We conclude that the strategy $\sigma$ is a winning strategy for $\mathsf{I}$.

    We now sketch how to prove the other half of the theorem.
    If $M,w\not\models\varphi$, then we can define a winning strategy for $\mathsf{II}$ similar to the strategy for $\mathsf{I}$ defined above.
    The main difference is that we need to consider $\mathsf{II}$-signatures, denoting approximants for $\mathsf{II}$'s variables.

    Let $\eta_1 X_1.\psi_1, \eta_2 X_2.\psi_2, \dots, \eta_n X_n.\psi_n$, enumerate the fixed-point subformulas of $\varphi$ in non-increasing size as above.
    We now enumerate the fixed-point subformulas of $\varphi$ which are owned by $\mathsf{II}$ in non-increasing size:
    \[
        \eta_1'' Z_1.\chi_1', \eta_2'' Z_2.\chi_2', \dots, \eta_k'' Z_k.\chi_k'.
    \]
    An \emph{$\mathsf{II}$-signature} $r = \tuple{r(1), \dots, r(m)}$ is a sequence of $k$ ordinals.
    As with $\mathsf{I}$-signatures, the lexicographical order is a well-ordering of the $\mathsf{II}$-signatures.

    Given a $\mathsf{II}$-signature $r$, we define augmented models $M^r_0, \dots, M^r_n$ by
    \begin{align*}
        M_0^r     &:= V; \\
        M_{i}^r &:= \left\{
        \begin{array}{ll}
            M_i[X_{i}\mapsto \|\eta_j'' Z_j^{r(j)}.\chi_j'\|^{M_i^r}], & \text{if $ X_{i} = Z_j$}; \\
            M_i[X_{i}\mapsto \|\eta_{i} X_{i}.\psi_{i}\|^{M_i}], & \text{if there is no $j$ such that $X_{i} = Z_j$}
        \end{array} \right. \text{, for $1\leq i \leq n$.}
    \end{align*}

    Again, if $M_n,v\models\psi$, we call $\tuple{v,\psi,\role}$ a true position; if $M_n,v \not\models\psi$, we call $\tuple{v,\psi,\role}$ a false position.
    Now, if $\tuple{v,\psi,\role}$ a true position, then there is a least signature $r$ such that $M_n^r,v\models\psi$.
    Similarly, if $\tuple{v,\psi,\role}$ a false position, then there is a least signature $r$ such that $M_n^r,v\not\models\psi$.
    Denote these signatures by $\mathsf{sig}^\mathsf{II}{v,\psi,\role}$.

    Similar to the first case, $\mathsf{I}$ cannot move in ways where the $\mathsf{II}$-signature increases, and we can build a strategy $\tau$ for $\mathsf{II}$ in a way that, eventually, their moves do not increase the signature.
    By the same argument as above, the strategy $\tau$ is winning.

\end{document}